\documentclass{llncs}
\usepackage[obeyDraft]{todonotes}
\usepackage[ruled]{algorithm}
\usepackage{algpseudocode}
\usepackage{makeidx}
\usepackage{graphicx}          
\graphicspath{{figures/}}

\usepackage{fainekos-macros}
\usepackage{cite}

\usepackage{ifthen} 
\newboolean{PEND}
\setboolean{PEND}{false}
\newboolean{REPORT}
\setboolean{REPORT}{true}

\newcommand{\unsafe}{\mathcal{U}}
\newcommand{\unspt}{u}
\renewcommand{\Re}{\mathbb{R}}

\title{Linear Hybrid System Falsification With Descent}

\title{Linear Hybrid System Falsification Through Descent\thanks{This work was partially supported by a grant from the NSF Industry/University Cooperative Research Center (I/UCRC) on Embedded Systems at Arizona State University and NSF award CNS-1017074.}}

\author{Houssam Abbas \and Georgios Fainekos}
\institute{Arizona State University, Tempe, AZ, USA,\\
\email{\{hyabbas,fainekos\}@asu.edu}\\ 
}

\begin{document}
	\maketitle
	
\begin{abstract}
	In this paper, we address the problem of local search for the falsification of hybrid automata with affine dynamics.
Namely, if we are given a sequence of locations and a maximum simulation time, we return the trajectory that comes the closest to the unsafe set.
In order to solve this problem, we formulate it as a differentiable optimization problem which we solve using Sequential Quadratic Programming.
The purpose of developing such a local search method is to combine it with high level stochastic optimization algorithms in order to falsify hybrid systems with complex discrete dynamics and high dimensional continuous spaces.
Experimental results indicate that indeed the local search procedure improves upon the results of pure stochastic optimization algorithms.	
\keywords{Model Validation and Analysis; Robustness; Simulation; Hybrid systems}
\end{abstract}


	
\section{Introduction}

Despite the recent advances in the computation of reachable sets in medium to large-sized linear systems (about 500 continuous variables)  \cite{girardG_IFAC2008,AsarinDMT10atva}, the verification of hybrid systems through the computation of the reachable state space remains a challenging problem \cite{guernicG_CAV2009,Althoff2010nonlinear}.
To overcome this difficult problem, many researchers have looked into testing methodologies as an alternative.
Testing methodologies can be coarsely divided into two categories: robust testing \cite{GirardP06hscc,JuliusFALP07hscc,DangDMS08cdc} and systematic/randomized testing \cite{BranickyCLM06iee,PlakuKV09tacas,RizkBFS08cmsb,ZulianiPC10hscc}.

Along the lines of randomized testing, we investigated the application of Monte Carlo techniques \cite{NghiemSFIGP10hscc} and metaheuristics to the temporal logic falsification problem of hybrid systems.
In detail, utilizing the robustness of temporal logic specifications \cite{FainekosP09tcs} as a cost function, we managed to convert a decision problem, i.e., does there exist a trajectory that falsifies the system, into an optimization problem, i.e., what is the trajectory with the minimum robustness value?
The resulting optimization problem is highly nonlinear and, in general, without any obvious structure.
When faced with such difficult optimization problems, one way to provide an answer is to utilize some stochastic optimization algorithm like Simulated Annealing.

In our previous work \cite{NghiemSFIGP10hscc}, we treated the model of the hybrid system as a black box since a global property, such as convexity of the cost function, cannot be obtained, in general.
One question that is immediately  raised is whether we can use ``local" information from the model of the system in order to provide some guidance to the stochastic optimization algorithm.

In this paper, we set the theoretical framework to provide local descent information to the stochastic optimization algorithm.
Here, by local we mean the convergence to a local optimal point.
In detail, we consider the falsification problem of affine dynamical systems and hybrid automata with affine dynamics where the uncertainty is in the initial conditions.
In this case, the falsification problem reduces to an optimization problem where we are trying to find the trajectory that comes the closest to the unsafe set (in general, such a trajectory is not unique).
A stochastic optimization algorithm for the falsification problem picks a point in the set of initial conditions, simulates the system for a bounded duration, computes the distance to the unsafe set and, then, decides on the next point in the set of initial conditions to try.
Our goal in this paper is to provide assistance at exactly this last step.
Namely, how do we pick the next point in the set of initial conditions?
Note that we are essentially looking for a descent direction for the cost function in the set of initial conditions.

Our main contribution, in this paper, is an algorithm that can propose such descent directions. Given a test trajectory $s_{x_0}: \Re_+ \mapsto \Re^n$ starting from a point $x_0$, the algorithm tries to find some vector $d$ such that $s_{x_0+d}$ gets closer to the unsafe set than $s_{x_0}$. We prove that it converges to a local minimum of the robustness function in the set of initial conditions, and demonstrate its advantages within a stochastic falsification algorithm. The results in this paper will enable local descent search for the satisfaction of arbitrary linear temporal logic specifications, not only safety specifications.

\section{Problem Formulation}

The results in this paper will focus on the model of hybrid automata with affine dynamics. 
A hybrid automaton is a mathematical model that captures systems that exhibit both discrete and continuous dynamics. 
In brief, a {\it hybrid automaton} is a tuple 
\[\Hc = (X, L, E, Inv, Flow, Guard, Re)\] 
where $X \subseteq \Re^n$ is the state space of the system, $L$ is the set of control locations, $E \subseteq L \times L$ is the set of control switches, $Inv : L \rightarrow 2^X$ assigns an invariant set to each location, $Flow : L \times X \rightarrow \Re^n$ defines the time derivative of the continuous part of the state, $Guard : E \rightarrow 2^X$ is the guard condition that enables a control switch $e$ and, finally, $Re : X \times E \rightarrow X \times L$ is a reset map. 
Finally, we let $H = L \times X$ to denote the state space of the hybrid automaton $\Hc$.

Formally, the semantics of a hybrid automaton are given in terms of generalized or timed transition systems \cite{Henzinger96}.
For the purposes of this paper, we define a {\it trajectory} $\eta_{h_0}$ starting from a point $h_0 \in H$ to be a function $\eta_{h_0} : \Re_+ \rightarrow H$.
In other words, the trajectory points to a pair of control location - continuous state vector for each point in time: $\eta_{h_0}(t) = (l(t), s_{x_0}(t))$, where $l(t)$ is the location at time $t$, and $s_{x_0}(t)$ is the continuous state at time $t$.
We will denote by $\loc(\eta_{h_0}) \in L^* \cup L^\omega$ the sequence of control locations that the trajectory $\eta_{h_0}$ visits (no repetitions).
The sequence is finite when we consider a compact time interval $[0,T]$ and $\eta$ is not Zeno.

{\bf Assumptions:} In the following, we make a number of assumptions.
First, we assume that for each location $v \in L$ the system dynamics are affine, i.e., $\dot x = Flow(v,x) = A x + b$, where $A$ and $b$ are matrices of appropriate dimensions. 
Second, we assume that the guards in a location are non-overlapping and that the transitions are taken as soon as possible.
Thirdly, we assume that the hybrid automaton is deterministic, i.e., starting from some initial state, there exists a unique trajectory $\eta_{h_0}$ of the automaton. 
This will permit us to use directly results from \cite{JuliusFALP07hscc}.
We also make the assumption that the simulation algorithms for hybrid systems are well behaved. 
That is, we assume that the numerical simulation returns a trajectory that remains close to the actual trajectory on a compact time interval.
To avoid a digression into unnecessary technicalities, we will assume that both the set of initial conditions and the unsafe set are included in a single (potentially different) control location.

Let $\unsafe \subseteq H$ be an unsafe set and let $D_\unsafe : H \mapsto \Re_+$ be the distance function to $\unsafe$, defined by
	\[ D_\unsafe(v,x) = \left \{
\begin{array}{cc} 
d_\unsafe(x) & \mbox{ if } v \in \proj_L(\unsafe) \\
+\infty & \mbox{ otherwise } 
\end{array} \right .
\]
where $\proj_L$ is the projection to the set of locations, $\proj_X$ is the projection to the continuous state-space and 
	\[ d_\unsafe(x) = \inf_{\unspt\in \unsafe} ||x-\unspt||. \]

\begin{defn} [Robustness]
Given a compact time interval $[0,T]$, we define the robustness of a system trajectory $\eta_{h}$ starting at some $h = (l,x)\in H$ to be $f(h) \triangleq \min_{0 \le t \le T} D_{\unsafe}(\eta_{h}(t))$. When $l$ is clear from the context, we'll write $f(x)$. 
\end{defn}

Our goal in this paper is to find operating conditions for the system which produce trajectories of minimal robustness, as they indicate potentially unsafe operation. This can be seen as a 2-stage problem: first, decide on a sequence of locations to be followed by the trajectory. Second, out of all trajectories following this sequence of locations, find the trajectory of minimal robustness. This paper addresses the second stage. The central step is the solution the following problem: 
\begin{prob}
Given a hybrid automaton $\Hc$, a compact time interval $[0,T]$, a set of initial conditions $H_0 \subseteq H$ and a point $h_0 = (l_0,x_0) \in H_0$ such that $0 < f(h_0) < +\infty$, find a vector $dx$ such that $h_0' = (l_0,x_0+dx)$, $\loc(
\eta_{h_0}) = \loc(\eta_{h_0'})$ and $f(h_0') \le f(h_0)$.
\label{pb:main}
\end{prob}

An efficient solution to Problem \ref{pb:main} may substantially increase the performance of the stochastic falsification algorithms by proposing search directions where the robustness decreases.
In summary, our contributions are:
\begin{itemize}
\item We formulate Problem \ref{pb:main} as a nonlinear optimization problem, which we prove to be differentiable w.r.t. the initial conditions. Thus it is solvable with standard optimizers.
\item We developed an algorithm, Algorithm~\ref{alg:red}, to find local minima of the robustness function.
\item We demonstrate the use of Algorithm~\ref{alg:red} in a higher-level stochastic falsification algorithm, and present experimental results to analyze its competitiveness against existing methods. 
\end{itemize}
	
		\section{Finding a descent direction}\label{sectionDescentDirection}

	Consider an affine dynamical system in $\Re^n$, 
	\[ \dot x = F(x)= Ax + b \]
	which we assume has a unique solution 
	\[s_{x_0}(t) = e^{At}x_0 + c(t) \]
	where $x_0 \in X_0$ is the initial state of the trajectory
				
	Let $\unsafe \subset \mathbb{R}^n$ be the \emph{convex} set of bad states, and $\overline{\unsafe}$ its closure. 
Note that even for linear systems, $f: X_0 \mapsto \Re_+$ is not necessarily differentiable or convex.
Our goal is to find the trajectory of minimum robustness. That is done by a local search over the set of initial conditions.
	
	Given an initial state $x_0$ and a trajectory $s_{x_0}$ that starts at $x_0$, define the time $t^*$ of its closest proximity to $\unsafe$, and the point $\unspt^* \in \overline{\unsafe}$ which is closest to the trajectory:
		
	\[t^* = arg\min_{t \geq 0} d_\unsafe(s_{x_0}(t)), \unspt^* = arg\min_{\unspt \in \overline{\unsafe}} ||s_{x_0}(t^*) - \unspt||\]

\subsection{Partial descent based at the nearest point}\label{subsecParDesc}	
	
	Given $t^*$, choose an \emph{approach vector} $d'$ such that $s_{x_0}(t^*) + d'$ is closer to $\unsafe$ than $s_{x_0}(t^*)$. Such a vector always exists given that $s_{x_0}$ has a positive distance to $\unsafe$. Moreover, it is not unique. Thus we have
	
	\[ f(x_0) = ||s_{x_0}(t^*) - \unspt^*|| > \min_\unspt ||s_{x_0}(t^*) +d' - \unspt|| \]
	
	Define $d = e^{-At^*}d'$. Then
	\begin{equation*}\label{eqDescent}
	\begin{split}
	f(x_0)& > \min_\unspt ||s_{x_0}(t^*) + d' - \unspt|| = \min_\unspt ||e^{At^*}x_0 + c(t) + e^{At^*}d - \unspt|| \\
				& \geq \min_t \min_\unspt ||(x_0 + d) e^{At} + c(t)  - \unspt|| = f(x_0 +d) \geq 0
	\end{split}
	\end{equation*}
	and $d$ is a descent direction, provided that $x_0 + d \in X_0$. \\
	
	It is easy to see that for any $x_0\in X_0$ and $d' \in \Re^n$, 
	\begin{equation*}\label{eqShift}
	s_{x_0+e^{-At^*}d'}(t) = s_{x_0}(t^*)+d'
	\end{equation*}
	so the new distance is achieved at the same time $t^*$ as the old one. This new distance $d_\unsafe(s_{x_0 + d}(t^*))$ is an upper \emph{bound} on the new trajectory's robustness. In general, the new trajectory's robustness might be even smaller, and achieved at some other time $t' \neq t^*$. 
	
	As pointed out earlier, the approach $d'$ is not unique. The requirement on $d'$ is that $s_{x_0}(t^*) + d'$ be closer to $\unsafe$ than $s_{x_0}(t^*)$. So define the set $P(x_0;t^*)$ of points that are closer to $\unsafe$ than $s_{x_0}(t^*)$ (see Fig. \ref{fig:PPc}):
	
	\begin{equation}\label{eqPDef}
	P(x_0; t^*) \triangleq \{x \in \mathbb{R}^n | d_\unsafe(x) \leq f(x_0)\} 
	\end{equation}
	
\begin{figure}[ht]
\centering
\includegraphics[width=120pt]{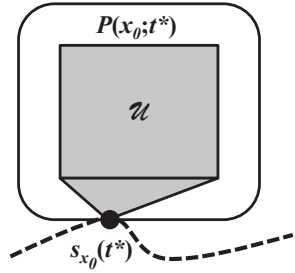}
\caption{The unsafe set $\unsafe$ and the set $P(x_0; t^*)$. The system trajectory $s_{x_0}$ appears as a dashed curve.}
\label{fig:PPc}
\end{figure}

	Then $d'$ must satisfy $s_{x_0}(t^*) + d' \in P(x_0; t^*) \Leftrightarrow d \in e^{-At^*}(P(x_0; t^*) - s_{x_0}(t^*))$. Combined with the requirement that $x_0 + d \in X_0$, we get
	
	\begin{equation*}\label{eqParDesc}
	d \in (X_0 - x_0) \bigcap e^{-At^*}[P(x_0;t^*) - s_{x_0}(t^*)] 
	\end{equation*}
	
	Any point in the above \emph{descent set} is a feasible descent direction. As a special case, it is easy to verify that $d = e^{-At^*}(\overline{\unspt} - s_{x_0}(t^*))$, for any $\overline{u} \in \overline{\unsafe}$, is a descent direction that leads to 0 robustness. 
	Coupled with the requirement that $x_0 + d$ must be in $X_0$, it comes
	\begin{equation*}\label{eqCompleteDesc}
	d \in  (X_0 - x_0) \bigcap e^{-At^*}(\overline{\unsafe} - s_{x_0}(t^*))
	\end{equation*}	 
	
	If computing $P$ is too hard, we can approximate it with the following $\unsafe^U$: 
	imagine translating $\unsafe$ along the direction $v = s_{x_0}(t^*) - \unspt^*$, so it is being drawn closer to $s_{x_0}(t^*)$, until it meets it. 
	Then we claim that the union of all these translates forms a set of points closer to $\unsafe$ than $s_{x_0}(t^*)$:
	
	\begin{proposition}\label{propBU}
	Let $\unsafe$ be a convex set, $s_{x_0}(t^*)$ a point outside it, and $\unsafe^U(v)$ be the Minkowski sum of $\unsafe$ and $\{\alpha v | \alpha \in [0,1]\}$.
	Then for any $p$ in $\unsafe^U(v)$, $d_\unsafe(p) \leq d_\unsafe(s_{x_0}(t^*))$
	\end{proposition}	
\ifthenelse{\boolean{REPORT}}
{	
	\begin{proof}
	$\unsafe^U$ is convex by the properties of Minkowski sums. 
	Let $\overline{\unspt} \in \partial \unsafe$. Then for any $\alpha \leq 1$, $d_\unsafe(\overline{\unspt}+\alpha v) \leq ||\overline{\unspt}+\alpha v - \overline{\unspt}|| = \alpha ||v|| = \alpha f(x_0) \leq f({x_0})$. 
	So translates of boundary points are closer to $\unsafe$ than $s_{x_0}(t^*)$.
	
	Now we show that all points in $\unsafe^U/\unsafe$ are translates of boundary points. Consider any point $p = \unspt + \alpha v$ in $\unsafe^U/\unsafe$: $\unspt$ is in $\unsafe$, but $p$ is not, so the line $[\unspt, p]$ crosses $\partial \unsafe$ for some value $\alpha^o$: $\unspt+\alpha^o v \in \partial \unsafe$. And, $p = \unspt+\alpha v = (\unspt+\alpha^o v) + (\alpha - \alpha^o)v$, so by what preceded, $d_\unsafe(p) \leq f(x)$.\\
	When $p \in \unsafe$, of course, $d_\unsafe(p) = 0 \leq f(x)$. 
	\end{proof}	
} 
{
	\begin{proof}
	See full-length technical report\cite{hya_gf_descentonline}.
	\end{proof}	
} 
	We have thus defined 3 possible descent sets: $\unsafe \subset \unsafe^U \subset P(x_0;t^*)$. 	
	\ifthenelse {\boolean{PEND}}
	{
	An alternative approximation to $P$ is the following set $P_c(x_0; t^*)$ (see Fig. \ref{fig:PPc}):

	\[P_c(x_0; t^*) \triangleq conv\{\unsafe,s_{x_0}(t^*)\} \]
	
	We show that $P_c \subset P$ by showing that all $p \in P_c$ satisfy $d_\unsafe(p) \leq d_\unsafe(s_{x_0})$ (drop the $t^*$ to avoid clutter). Let $p\in P_c$: then there exists a $\lambda \in [0,1] $ and a $\unspt \in \unsafe$ s.t.
	\begin{equation*}
	\begin{split}
	& p = \lambda s_{x_0} + (1-\lambda)\unspt\\
	& d_\unsafe(\lambda s_{x_0} + (1-\lambda)\unspt) \leq \lambda d_\unsafe(s_{x_0})+(1-\lambda)d_\unsafe(\unspt) = \\
    & = \lambda d_\unsafe(s_{x_0}) \leq d_\unsafe(s_{x_0})
	\end{split}
	\end{equation*}

	($d_\unsafe$ is convex if $\unsafe$ is convex). Thus we may choose
	\begin{equation}\label{eqParDescPc}
	d \in (X_0 - x_0) \bigcap e^{-At^*}[P_c(x_0; t^*) - s_{x_0}(t^*)]
	\end{equation}
	
	Any point in the set above is a descent direction. This is also a general result:
	
	\begin{proposition}
	Given any $y \in \Re^n / \unsafe$, points in $conv\{y,\unsafe\}$ are closer to $\unsafe$ than $y$.
	\end{proposition}
	
	Note the following inclusion:
	\[P_c(x_0, t^*) \subseteq \unsafe^U(s_{x_0}(t^*) - \unspt^*) \subseteq P(x_0,t^*)\]	
	
	If $x_0+d \notin X_0$, then the following proposition shows that it can be shortened to fit in $X_0$, and still produce a descent:
	\begin{proposition}\label{lemmaApproachCircle}
	Let $s_x(t) + d' \in P_c(x;t)$, and $d = e^{-At^*}d'$. Then $\mu d$ is a descent vector for all $\mu \in (0,1]$
	\end{proposition}
	
	\begin{proof}
	$s_x(t) + \mu d'$ is on the line $(s_x(t), s_x(t) + d']$, and so is a convex combination of  $s_x(t), s_x(t) + d' \in P_c$, thus is itself in $P_c$. Therefore $d_\unsafe(s_x(t)+\mu d') < d_\unsafe(s_x(t))$. Therefore $\mu e^{-At^*}d'$ is a descent vector.
	\end{proof}
	
	}		{}  

	\subsection{Implementation}\label{subsecImpl}
	The question we address here is: how do we obtain, \emph{computationally}, points in the descent set $\Wc$, where $\Wc = \unsafe, P(x_0)$ or $\unsafe^U(v)$? The following discussion is based on Chapters 8 and 11 of \cite{BoydV_book04}.
	
	Since we're assuming $X_0$ and $\unsafe$ to be convex, then the descent set is also convex.  
	Describe $X_0$ with a set of $N_X$ inequalities $q_i(x) \leq 0$ where the $q_i$ are convex and differentiable, and $\Wc = \{x | p_i(x;x_0) \leq 0, i=1...k\}$ for convex differentiable $p_i$ (the particular form of the $p_i$ will depend on the descent set at hand). We assume dom $p_i$ = dom $q_i \triangleq \Re^n$. 	
	
	\emph{Given an already simulated trajectory $s_{x_0}$} and its time of minimum robustness $t^*$, we are looking for a feasible $x_1$ such that $s_{x_1}(t) \in \Wc$ for some $t$. Thus we want to solve the following feasibility problem\\
	\begin{equation}\label{eqPCP}
	\begin{split}
	\min_{(x,\nu)}\phantom{11} &\nu \\
	s.t.\phantom{11} &p_i(s_x(t);x_0) \leq \nu, i=1\ldots k\phantom{1111}\textrm{(t-PDP($x_0$))}\\
	\phantom{11} &q_i(x) \leq \nu, i=1\ldots N_X\\	     
	\end{split}
	\end{equation}
	
	This is a convex program, which can be solved by a Phase I Interior Point method~\cite{BoydV_book04}. A non-positive minimum $\nu^*$ means we found a feasible $x$; if $\Wc = \unsafe$, then our work is done: we have found an unsafe point. Else, we can't just stop upon finding a non-positive minimum: we have merely found a new point $x_1$ whose robustness is less than $x_0$'s, but not (necessarily) 0. So we iterate: solve t-PDP$(x_0)$ to get $x_1$, solve t-PDP$(x_1)$ to get $x_2$, and so on, until $f(x_i) = 0$, a maximum number of iterations is reached, or the problem is unsolvable. If the minimum is positive, this means that \emph{for this value of $t$}, it is not possible for any trajectory to enter $\overline{\unsafe}$ at time $t$. 	
	
	The program suffers from an arbitrary choice of $t$. One approach is to sample the trajectory at a fixed number of times, and solve \eqref{eqPCP} for each. This is used in the experiments of this section. A second approach, used in the next section, is to let the optimization itself choose the time, by adding it to the optimization variable. 
	The resulting program is no longer necessarily convex.

\ifthenelse{\boolean{REPORT}} {	
		\subsection{Numerical Experiments}\label{sectionLinNumer}
	
	In this section, we present some numerical experiments demonstrating the practical significance of the previous theoretical results.

\ifthenelse{\boolean{PEND}} 
{
\begin{exmp}
We first present a simple 2D example (Example 7.15 in Tabuada \cite{Tabuada2009}).
The linear system has the following dynamics
\[ \dot x = 
\begin{bmatrix}
 -7 & -3 \\
 3  & 0 
\end{bmatrix}
x.
\]
We consider two unsafe sets: $\unsafe_1 = \{ x \in \Re^2 \; | \; x_2 \ge 4 \}$ and $\unsafe_2 = \{ x \in \Re^2 \; | \; x_2 \ge 2 \}$.
The set $\unsafe_1$ is not reachable, while $\unsafe_2$ is.
In both cases, we apply the descent algorithm starting from a trajectory at $[5.5 \; 0]^T$.
In the first case (Fig. \ref{fig:2d:01}), the algorithm does not find a complete descent direction and, thus, attempts to find a partial descent direction.
The algorithm terminates and returns the trajectory with minimum robustness that is found.
In the second case (Fig. \ref{fig:2d:02}), the unsafe set is reachable and the algorithm finds a trajectory that falsifies the specification.

\begin{figure}[ht]
\centering
\includegraphics[width=200pt]{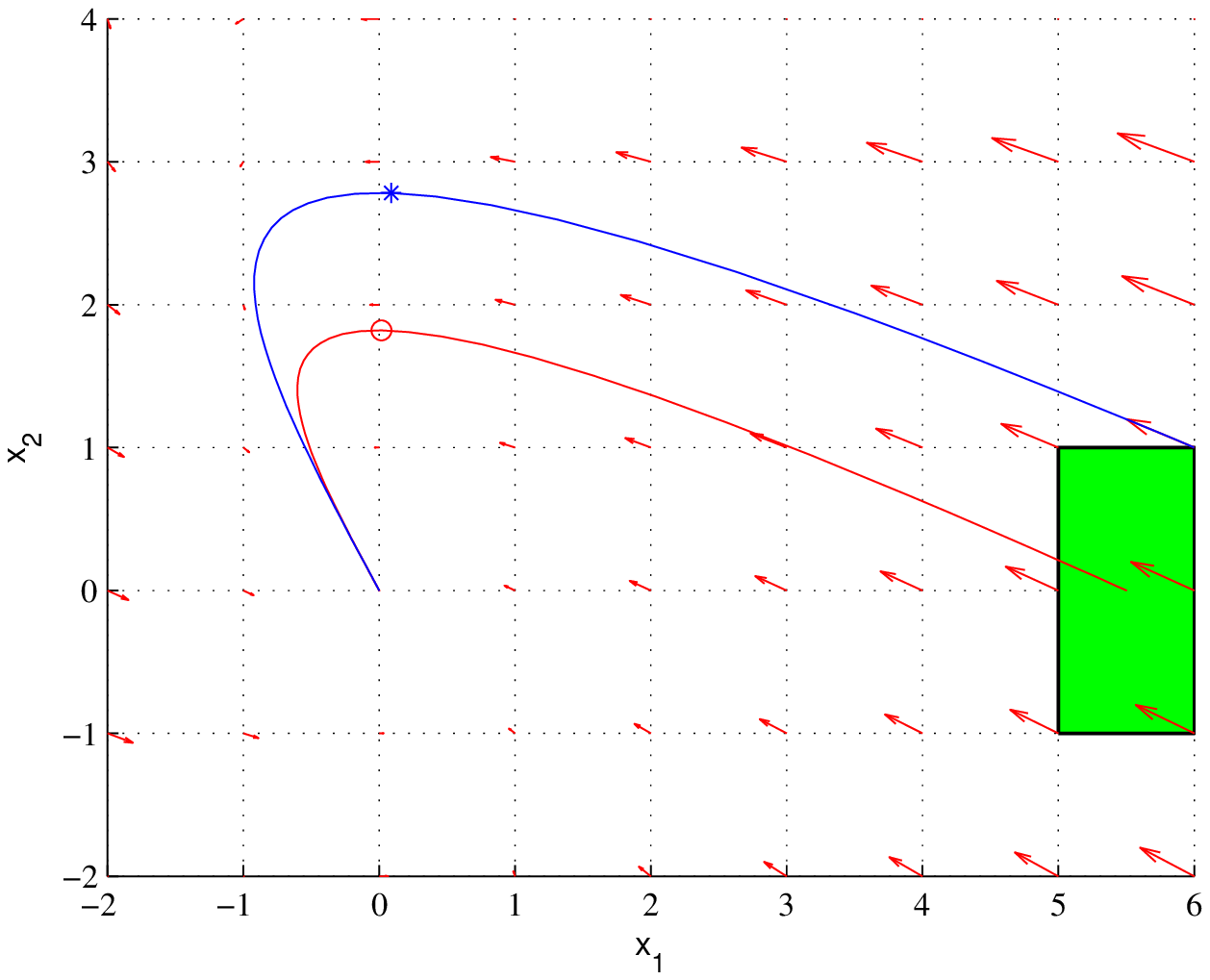}
\caption{[Unreachable unsafe set] The initial test trajectory starting from $[5.5 \; 0]^T$ and the trajectory found with minimum robustness. 
The circle indicates the point of minimum robustness for the initial trajectory, while the star the point of minimum robustness for the trajectory resulting after the descent.}
\label{fig:2d:01}
\end{figure}

\begin{figure}[ht]
\centering
\includegraphics[width=200pt]{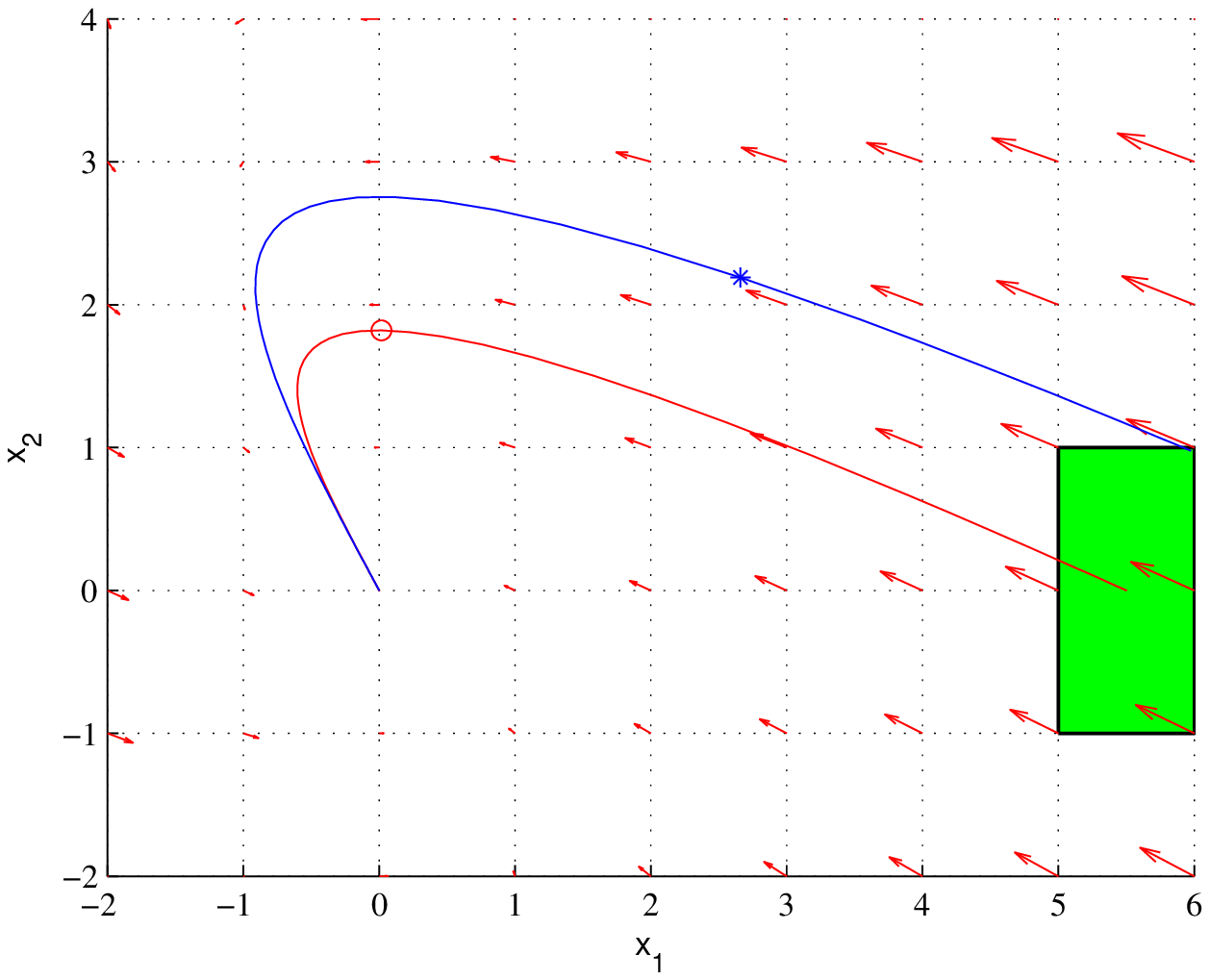}
\caption{[Reachable unsafe set] The initial test trajectory starting from $[5.5 \; 0]^T$ and the trajectory found with minimum robustness. 
The circle indicates the point of minimum robustness for the initial trajectory, while the star the point of almost zero robustness that the descent algorithm found.}
\label{fig:2d:02}
\end{figure}

\end{exmp}
}
{}

\begin{exmp} \label{exmp:rlc}
We consider the verification problem of a transmission line \cite{Han2005Phd}. 
The goal is to check that the transient behavior of a long transmission line has acceptable overshoot for a wide range of initial conditions. 
Figure \ref{fig:rlc} shows a model of the transmission line, which consists of a number of RLC components (R: resistor, L: inductor and C: capacitor) modeling segments of the line.
The left side is the sending end and the right side is the receiving end of the transmission line. 

\begin{figure}[ht]
\begin{center}
\includegraphics[width=200pt]{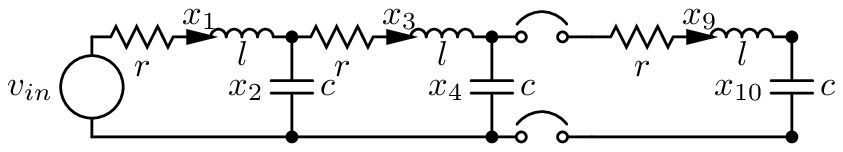}
\caption{RLC model of a transmission line.} 
\label{fig:rlc}
\end{center}
\end{figure}

The dynamics of the system are given by a linear dynamical system
$$
\dot x(t) = Ax(t)+b V_{in}(t) \mbox{ and } V_{out}(t) = C x(t)
$$
where $x(t)\in \Re^{81}$ is the state vector containing the voltage of the capacitors and the current of the inductors and $V_{in}(t)\in \Re $ is the voltage at the sending end. 
The output of the system is the voltage $V_{out}(t) \in \Re$ at the receiving end. 
Here, $A$, $b$ and $C$ are matrices of appropriate dimensions. 
Initially, we assume that the system might be in any operating condition such that $x(0) \in [-0.1,0.1]^{41} \times [-0.01,0.01]^{40}$.
Then, at time $t=0$ the input is set to the value $V_{in}(t)=1$.

The descent algorithm is applied to the test trajectory that starts from $x(0) = 0$ and it successfully returns a trajectory that falsifies the system (see Fig. \ref{fig:fals}).

\begin{figure}[ht]
\begin{center}
\includegraphics[width=200pt]{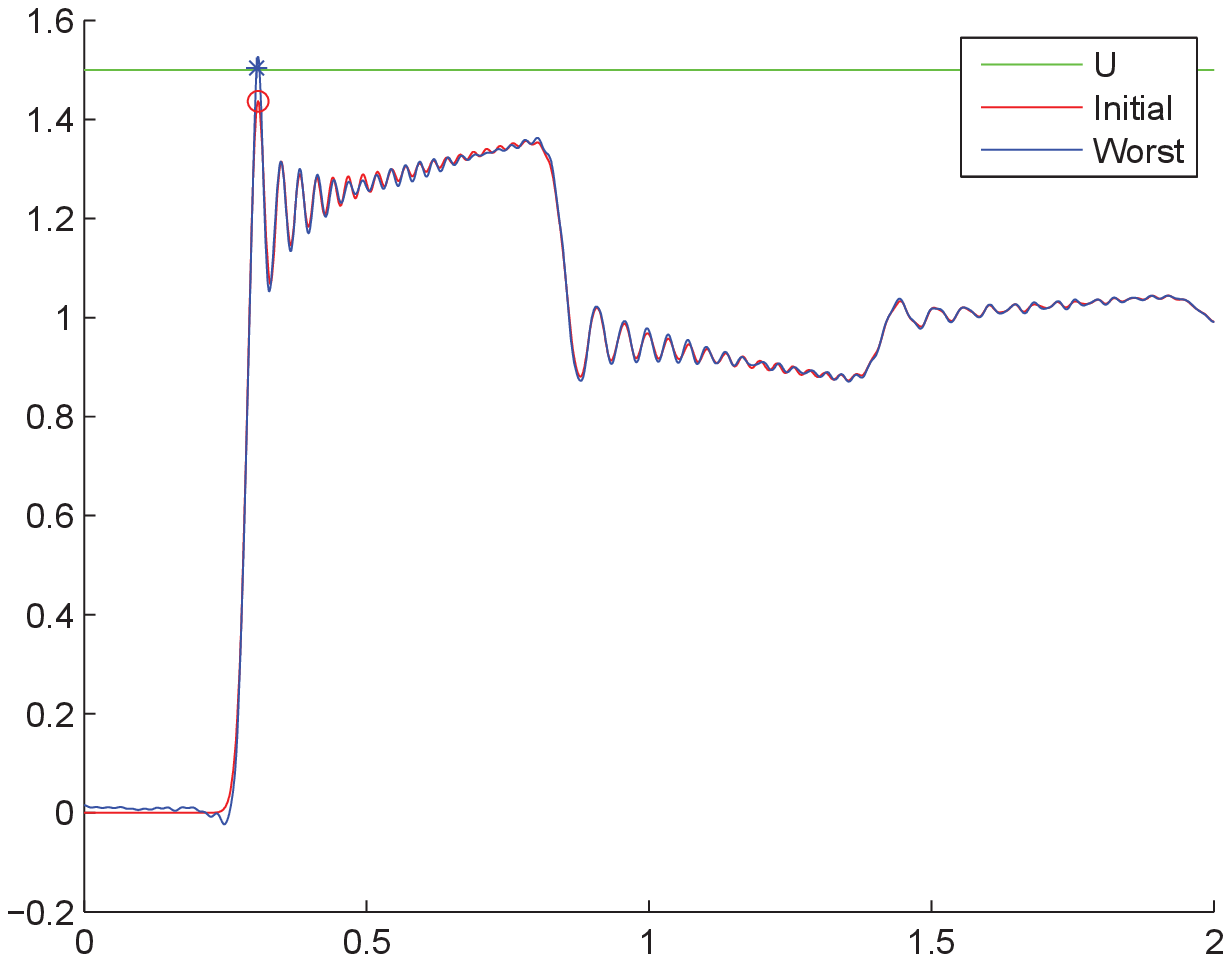}
\caption{The unsafe set $U$, the initial test trajectory starting from $x(0) = 0$ and the trajectory that falsifies the system.} 
\label{fig:fals}
\end{center}
\end{figure}

\end{exmp}
		
	}
	{}
	
	\section{Hybrid systems with affine dynamics}\label{sectionHybrid}

	We now turn to the case of hybrid systems with affine dynamics in each location. 
	The objective is still to find a descent direction in $H_0$, given a simulated trajectory $\eta_{h_0}$ originating at point $h_0 \in H_0$. 
	Note that since we have assumed that $\proj_L(H_0)$ is a singleton set, the problem reduces to finding a descent direction in $X_0 = \proj_X(H_0)$.
	
	{\bf Assumptions}. At this point, we make the following assumptions:
	
	a. The continuous dynamics in each location are stable.\footnote{This is not a restrictive assumption since we can also consider incrementally stable systems~\cite{Tabuada2009}, and even unstable linear systems\cite{JuliusP09hscc}. }
	
	b. For every transition $e \in L^2$, the resets $Re(\cdot,e)$ are differentiable functions of their first argument.
	
	c. Conditions 4 and 5 of Theorem III.2 in~\cite{LygerosJSZS03tac} are satisfied, namely: for all $i$, there exists a differentiable function $\sigma_i: \Re^n \mapsto \Re$ such that $Inv(l_i)=\{x \in \Re^n | \sigma_i(x)\geq 0\}$; and, for all $i,x$ such that $\sigma_i(x)=0$, the Lie derivative $L_F \sigma_i(x) \ne 0$. 
This allows us to have a \emph{differentiable} transition time $t_x$ of the trajectory starting at the initial point $x\in X_0$.
	
	d. The sequence of locations $loc(\eta_{h_0})$ enters the location of the unsafe set. This is required for our problem to be well-defined (specifically, for the objective function to have finite values). The task of finding such an $h_0$ is delegated to the higher-level stochastic search algorithm, within which our method is integrated.
	
	\subsection{Descent in the Robustness Ellipsoid}\label{subsecRED}			
 	Consider a trajectory $\eta_{h_0}$ with positive robustness, with $\loc(\eta_{h_0}) = l_0 l_1 \dots l_N$. This is provided by the simulation. 
 	Let the initial set $X_0$ be in location $l_0$ and let $l_\unsafe$ denote the location of $\unsafe$. In order to solve Problem \ref{pb:main}, we assume that $l_\unsafe$ appears in  $\loc(\eta_{h_0})$ (see Assumption d above) - otherwise, $f(h_0) = +\infty$ and the problem as posed here is ill-defined.
 	We search for an initial point $h_0' \in H_0$ (actually $x_0' \in X_0$), whose trajectory gets closer to the unsafe set than the current trajectory $\eta_{h_0}$. 			
		
	In order to satisfy the constraints of Problem \ref{pb:main}, we need to make sure that the new point $h_0'$ that we propose generates a trajectory that follows the same sequence of locations as $\eta_{h_0}$.
This constraint can be satisfied using the notion of robust neighborhoods introduced in \cite{JuliusFALP07hscc}.
In \cite{JuliusFALP07hscc}, it is shown that for stable systems and for a given safe initial point $h_0 = (l_0,x_0)$, there exists an `ellipsoid of robustness' centered on $x_0$, such that any trajectory starting in the ellipsoid, remains in a tube around $\eta_{h_0}$. The tube has the property that all trajectories in it follow the same sequence of locations as $\eta_{h_0}$. Therefore, we restrict the choice of initial point to $X_0 \bigcap E(x_0)$, where $E(y) = \{x | (x-y)^TR^{-1}(x-y) \leq 1\}$ is the ellipsoid of robustness centered on $x_0$, with shape matrix $R$. 
Formally, in \cite{JuliusFALP07hscc}, the following result was proven.
	
\begin{thm}
Consider a hybrid automaton $\Hc$, a compact time interval $[0,T]$, a set of initial conditions $H_0 \subseteq H$ and a point $h_0 = (l_0,x_0) \in H_0$. 
Then, we can compute a number $\varepsilon >0$ and a bisimulation function $\phi(x_1,x_2) = (x_1 - x_2)^T M (x_1-x_2)$, where $M$ is a positive semidefinite matrix, such that for any $x_0' \in \{ y \in X \; | \; \phi(x_0,y) \le \varepsilon \}$, we have $\loc(\eta_{h_0}) = \loc(\eta_{(l_0,x_0')})$. 
\label{thm:bis}
\end{thm}

\begin{rem}
(i) In \cite{JuliusFALP07hscc}, in the computation of $\varepsilon$, we also make sure that any point in the robust neighborhood generates a trajectory that does not enter the unsafe set.
In this work, we relax this restriction since our goal is to find a point that generates a trajectory that might enter the unsafe set.
(ii) In view of Theorem \ref{thm:bis}, the shape matrix for the ellipsoid is defined as $R = \varepsilon^2 M^{-1}$.
\end{rem}

We now proceed to pose our search problem as a feasibility problem.	Let $t_0$ be the time at which $s_{x_0}$ is closest to $\unsafe$. 
We choose $P(x_0;t_0)$ as our descent set: recall that it is the set of all points which are closer to $\unsafe$ than $s_{x_0}(t_0)$ (Def. \ref{eqPDef}). 
Therefore, if we can find $x^* \in X_0 \bigcap E(x_0)$ such that $s_{x^*}(t^*) \in P(x_0;t_0)$ for some $t^*$, it follows that $f(x^*) \leq f(x_0)$. 
To simplify notation, let $\Wc =P(x_0;t_0)$ be the descent set. As before, it is assumed that $\Wc = \{x \in \Re^n : p_i(x) \leq 0, i=1\ldots k\}$ for differentiable $p_i$.
The search problem becomes:

Given $\eta_{h_0}$, find $x^* \in X_0 \bigcap E(x_0)$ and $t^* \geq 0$, such that $s_{x^*}(t^*) \in \Wc$. This is cast as an optimization problem over $z \in \Re^{n}\times\Re_+\times\Re$:
	\begin{equation} \label{eqHybridProg}
	\begin{split}
	\min_{z = (x,t,\nu)} & \nu \\	
	s.t.\phantom{11} & C_0 x - g_0 \leq 0\\
									 & (x-x_0)^TP^{-1}(x-x_0) -1 \leq \nu\\
									 & p_i(s_x(t);x_0)\leq \nu, i = 1\ldots k\\									
	\end{split}
	\end{equation}
where $s_x(t) = \proj_X(\eta_{(l_0,x)}(t))$ and $X_0 = \{x | C_0 x - g_0 \leq 0\}$.

\begin{rem}
Note that Problem \eqref{eqHybridProg} is specific to a choice of initial point $x_0$; this will be important in what follows. In our implementation, the first constraint is specified as bounds to the optimization and so is always satisfied.
\end{rem}

Later in this section, we discuss how to solve this optimization problem. 
For now, we show how solving this problem produces a descent direction for the robustness function. 
For convenience, for $z=(x,t,\nu)$, we define the constraint functions
	 \begin{subequations}\label{eq:cnstrgrp}
	\begin{align} 
	G_0(z) & = C_0 x - g_0 \label{eq:G0}\\
	G_E(z) & = (x-x_0)^TP^{-1}(x-x_0) -1 \label{eq:GE}\\
	G_\Wc(z) & = \begin{pmatrix}									  
									  p_1(s_x(t);x_0)\\										  
									  \vdots \\
									  p_k(s_x(t);x_0) \label{eq:GW}\\
						 \end{pmatrix}
	\end{align}				
	\end{subequations}    
  A point $z$ is \emph{feasible} if it satisfies the constraints in Problem \eqref{eqHybridProg}. 
  Finally, define the objective function $F(z) = \nu$.
    
  The objective function $F(z)$ measures the \emph{slack} in satisfying the constraints: a negative $\nu$ means all constraints are strictly satisfied, and in particular, $G_\Wc$. 
Thus, we have a trajectory that enters $\Wc$ and, hence, gets strictly closer to $\unsafe$. This reasoning is formalized in the following proposition:

  \begin{proposition}\label{prop:hybDesc}
	Let $z^* = (x^*, t^*, \nu^*)$ be a minimum of $F(z)$ in program (\ref{eqHybridProg}). Then $f(l_0,x^*) \leq f(l_0,x_0)$.
	\end{proposition}

\ifthenelse{\boolean{REPORT}}
{	
	\begin{proof}
It is assumed that the optimizer is iterative and that it returns a solution that decreases the objective function. In what follows, for a vector $y \in \Re^n$, $\max y$ is the largest entry in $y$.

We first remark that for a given $x$ and $t$ that satisfy the constraints in \eqref{eqHybridProg}, 
\[z = (x,t,\max\{G_E(x,t), G_\Wc(x,t)\})\]
is feasible, and $F(z) \leq F(x,t,\nu)$ for any feasible $(x,t,\nu)$. Therefore, we may only consider points with $F(z) = \nu = \max\{G_E(x,t), G_\Wc(x,t)\}$.

Let $z_0 = (x_0, t_0, \nu_0)$ be the initial point of the optimization. 
Because $x_0$ is the center of $E(x_0)$, $G_E(z_0) = -1$. 
And, because $s_{x_0}(t_0) \in \partial\Wc$, $\max G_\Wc(z_0) = 0$. Thus $\nu_0 = 0$. Therefore, at the minimum $z^* = (x^*, t^*, \nu^*)$ returned by the optimizer, $\nu^* \leq \nu_0 = 0$. In particular, $G_\Wc(z^*) \leq 0$, and the new trajectory $s_{x^*}$ enters $\Wc$. Therefore, its robustness is no larger than that of the initial trajectory $s_{x_0}$. 
\end{proof}


We now address how Problem \ref{eqHybridProg} might be solved. Functions $F$, $G_0$ and $G_E$ are differentiable in $z = (x,t,\nu)$. It is not clear that $G_\Wc$, or equivalently, $p_i(s_x(t);x_0)$, \emph{as a function of $z$}, is differentiable. We now show that under some asumptions on the $p_i$, for trajectories of linear systems, $p_i$ is in fact differentiable in both $x$ and $t$, over an appropriate range of $t$. This implies differentiability in $z$. Therefore, standard gradient-based optimizers can be used to solve Problem \ref{eqHybridProg}. 

For the remainder of this section, we will re-write $s_x(t)$ as $s(x,t)$ to emphasize the dependence on the initial point $x$. 
$s^{(i)}(x,\tau)$ will denote the point, at time $\tau$, on the trajectory starting at $x \in Inv(l_i)$, and evolving according to the dynamics of location $i$. 
When appearing inside location-specific trajectories such as $s^{(i)}$, the time variable will be denoted by the greek letter $\tau$ to indicate \emph{relative} time: that is, time measured from the moment the trajectory entered $l_i$, not from datum $0$. $s(x,t)$ (without superscript) will denote the hybrid trajectory, traversing one or more locations. 
We will also drop the $x_0$ from $p_i(y;x_0)$, and write it simply as $p_i(y)$. 

We first prove differentiability in $x$. Therefore, unless explicitly stated otherwise, the term `differentiable' will mean `differentiable in $x$'. Start by noting that $p_i(s(x,t))$ is a composite function of $x \in X_0$. Since $p_i$ is differentiable, it is sufficient to prove that $s(x,t)$ is differentiable. 
The hybrid trajectory $s(x,\cdot)$ is itself the result of composing the dynamics from the visited locations $l_0,...,l_{N-1}$.

Recall that $E(x_0)$ is the ellipsoid of robustness centered at $x_0$. As shown by Julius et al.\cite{JuliusFALP07hscc}, the following times are well-defined:
\begin{defn} [Transition times]
Given $x_0 \in X_0$, let $E_0 \triangleq \mathrm{int}(E(x_0) \bigcap X_0)$.
$t_i$ is the time at which trajectory $s(x_0)$ transitions from $Inv(l_{i-1})$ into $Inv(l_i)$ through guard $Guard(l_{i-1},l_i)$. \\
$t_i^-$ is the maximal time for which the image of $E_0$ under the hybrid dynamics is contained in $Inv(l_{i-1})$:
\[t_i^- = \max\{t | s(E_0, t) \subset Inv(l_{i-1}) \}\] 
In other words, $t_i^-$ is the time at which occurs the first $l_{i-1}$-to-$l_i$ transition of a point in $s(E_0)$.

$t_i^+$ is the minimal time for which the image of $E_0$ under the hybrid dynamics is contained in $Inv(l_i)$:
\[t_i^+ = \min\{t | s(E_0, t) \subset Inv(l_i) \}\] 
In other words, $t_i^+$ is the time at which occurs the last $l_{i-1}$-to-$l_i$ transition of a point in $s(E_0)$.

For a given point $ x \in X_0$, $t_x^{i-1\rightarrow i}$ ($\tau_x^{i-1\rightarrow i}$) is the absolute (relative) transition time of trajectory $s^{(i)}(x)$ from $Inv(l_{i-1})$ into $Inv(l_i)$ through guard $Guard(l_{i-1},l_i)$. Thus, for example, $t_1 = t_{x_0}^{0\rightarrow 1} = \tau_{x_0}^{0\rightarrow 1}$ and $t_2 = t_{x_0}^{1\rightarrow 2} = \tau_{x_0}^{0\rightarrow 1}+\tau_{y_0}^{1\rightarrow 2}$, with $y_0=s^{0}(x_0,t_{x_0}^{0\rightarrow 1})$. When the transition is clear from context, we will simply write $t_x$ ($\tau_x$).

\end{defn}

We will first show differentiability of a trajectory that visits only 2 locations $l_0$ and $l_1$:
\begin{equation}\label{eq:s2locs}
s(x_0,t) = s^{(1)}(Re(s^{(0)}(x_0,t_x), (l_0,l_1)),t-t_x)
\end{equation}

\begin{exmp}
We first present a simple 1D example to illustrate the definitions and the idea of the proof. Consider the hybrid system with three locations 
\[\Hc = (\Re, \{0,1,2\}, \{(0,1), (1,2)\}, Inv, Flow, Guard, Id) \]
where $Inv(l) = \Re$ for $l=0,1,2$, and the flow is defined by 
\[ Flow(l,x) = \dot x(t) \left \{
\begin{array}{cc} 
x(t) & \mbox{ if } l \in \{0,2\} \\
-x(t)& \mbox{ if } l = 1
\end{array} \right .
\]
The guards are $Guard(0,1) = \{1\}$ and $Guard(1,2) = \{1/4\}$. $Id$ is the identity map, so there are no resets. The initial set is $X_0 = [0, 1/2]$. The solutions in the individual locations are then
\[s^{(0)}(x,t) = e^tx \]
\[s^{(1)}(x,t) = e^{-t}x\]
\[s^{(2)}(x,t) = e^tx\]

We can solve, in this simple case, for $\tau_x^{0\rightarrow 1}$: $e^{\tau_x}x=1 \Rightarrow \tau_x^{0\rightarrow 1} = \ln(1/x)$.
Similarly for $\tau_x^{1\rightarrow 2}$: $e^{-\tau_x}\cdot1 = 1/4\Rightarrow \tau_x^{1\rightarrow 2} = ln(4x)$.

We first show differentiability of the trajectory over locations 0 and 1. We then do the same for a trajectory over locations 1 and 2. Then we stitch the two together and show differentiability over 3 locations. For locations 0 and 1:
$s(x,t) = s^{(1)}(s^{(0)}(x,t_x),t-t_x) = s^{(1)}(1,t-t_x) =e^{-(t-t_x)}\cdot 1=e^{-t}/x \Rightarrow \frac{d}{dt} s(x,t) = -\frac{e^{-t}}{x^2}$. 

Moving on the trajectory over locations 1 and 2, the procedure is the same: from an initial point $x \in Guard(0,1) = \{1\}$,
for a fixed (relative time) $\tau \in (t_2-t_1, t_3^- - t_1)$:
$s(x,\tau) = s^{(2)}(s^{(1)}(x,\tau_x),\tau-\tau_x) = s^{(2)}(1/4,\tau+\ln(1/4x)) =e^{\tau+\ln(1/4x)} 1/4=e^\tau/16x \Rightarrow \frac{d}{d\tau} s(x,\tau) = -\frac{e^{\tau}}{16x^2}$. 

Finally we stitch up the 2 portions of the trajectory:
$x\in X_0$, $t \in [t_2, t_3^-]$. $s(x,t) = s^{(2)}(s^{(1)}(s^{(0)}(x,t_1),t_2-t_1),t-t_2) = s^{(2)}(s^{(1)}(1,t_2-t_1),t-t_2) = s^{(2)}(1/4,t-t_2) = e^{t-t_2}/4$. Since $t_2 = t_x^{0\rightarrow 1}+\tau_1^{1\rightarrow 2} = \ln(1/x) + \ln(4\cdot 1) = \ln(4/x) \Rightarrow s(x,t) = \frac{e^t}{4} e^{\ln(x/4)} = xe^t/16 \Rightarrow \frac{d}{dt} s(x,t) = e^t/16$.

\end{exmp}

We now prove the general case.
  \begin{proposition}
	Let $x_0 \in E_0$, and fix $t \in (t_1 ,t_2^-]$. Consider the hybrid trajectory over 2 locations in Eq.\eqref{eq:s2locs}. If Assumptions a-d are satisfied, then $s(x,t)$ is differentiable at $x_0$.
	\label{prop:diffProb13x}
	\end{proposition}
	
	\begin{proof}
	In what follows, $e = (l_0,l_1)$. 
	
\begin{equation*}\label{eq:Res0_diff}
\begin{split}
s^{(0)}(x,\tau_x)&=e^{\tau_x A_0}x+\int_{0}^{\tau_x}{e^{(\tau_x-s)A_0}bds} \\
&=\underbrace{e^{\tau_x A_0}x}_{term1}+ \underbrace{e^{\tau_x A_0}}_{term2}\underbrace{\int_{0}^{\tau_x}{e^{-sA_0}bds}}_{term3} \\
\end{split}
\end{equation*}	
Terms 1 and 2 are clearly differentiable in $x$. 
For term3, write $M(t) = \int_{0}^{t}{e^{-sA_0}bds}$ so term3 = $M(\tau_x)$. $M(t)$ is differentiable by the $2^{nd}$ Fundamental Theorem of Calculus and its derivative is $M'(t) = e^{-t A_0}b$. As a consequence of Assumption c, $\tau_x$ itself is differentiable in $x$ (Lemma III.3 in ~\cite{LygerosJSZS03tac}), and the chain rule allows us to conclude that term3 is differentiable in $x$. 
Thus $s^{(0)}(x,\tau_x)$ is differentiable over $E_0$. Since $Re(\cdot,e)$ is differentiable by Assumption b, then $Re(s^{(0)}(x,\tau_x),e)$ is differentiable over $E_0$. 
Note that $E_0$ is open and $s^{(0)}$ is continuous, so $U = \{w \in \Re^n | w = s^{(0)}(x,t_x) \mbox { for some } x\in E_0\} \subset Guard(e)$ is open. Since $Re(\cdot,e)$ is continuous, then $Re(U,e)$ is open. 
Next,

\begin{equation*}\label{eq:traj_2_locs}
\begin{split}
	s(x,t) &= s^{(1)}(Re(s^{(0)}(x,t_x), e),t-t_x)\\
	&=\underbrace{e^{(t-t_x)A_1}}_{term4}Re(s^{(0)}(x,t_x),e)+\underbrace{e^{(t-t_x)A_1}}_{term5} \underbrace{\int_0^{t-t_x}{e^{-s A_1}b_1ds}}_{term6}	
\end{split}
\end{equation*}	
Using the same argument as above, terms 4, 5 and 6 are differentiable in $x$. In conclusion, $s(x,t)$ is differentiable at over $E_0$, and this ends the proof.

\end{proof}

The following proposition generalizes Prop. \ref{prop:diffProb13x} to trajectories over more than 2 locations.

\begin{proposition} \label{prop:diffProb13x_manylocs}
Fix $t \in (t_{N-1} ,T]$, and consider the hybrid trajectory over $N \geq 1$ locations. Then $s(x,t)$ is differentiable at $x_0$ for all $x_0 \in E_0$.
\end{proposition}
\begin{proof}
We argue by induction over the number of locations $N$. The base case $N=1$ is true by hypothesis, and the case $N=2$ has been proven in Prop. \ref{prop:diffProb13x}. 
For $N>2$ and $t\leq t_{N-1}$, let $\zeta(x,t)$ be the trajectory over the first $N-1$ locations, so that $s(x,t) = s^{(N-1)}(Re(\zeta(x,\tau_{N-2}),(l_{N-2},l_{N-1})),t-t_{N-1})$. 
By the induction hypothesis, $\zeta(x,t)$ is differentiable at $x_0$. Then $\zeta$ and $s^{(N-1)}$ satisfy the conditions of the case $N=2$. 
\end{proof}

\todo[inline]{What about trajectories ending at $t=t_1$? then no, trajectory is not time-differentiable there - in fact, the trajectory is a multi-function at the transition times, so differentiability is not even a meaningful concept there.}
Differentiability with respect to time is easily proven: 

\begin{proposition}\label{prop:diffProb13t}
Let $x_0 \in E_0$ and $t \in (t_{N-1} , T)$, that is, a time at which the trajectory is in the last location. Consider the hybrid trajectory over $N \geq 1$ locations. Then $s(x_0,t)$ is differentiable in $t$ over $[t_{N-1}, T)$.
	\end{proposition}
	\begin{proof}
	$s(x_0,t) = s^{(N-1)}(x_0,t-t_{N-1})$. The location-specific trajectories $s^{(i)}(x,\cdot)$ are solutions of differential equations involving at least the first time derivative. Therefore, they are smooth over $(t_{N-1},T)$. This implies differentibility of the hybrid trajectory $s(x_0,\cdot)$ over the same interval. At $t=T$, the trajectory is only left-differentiable, since it's undefined from the right.
	\end{proof}
	
The following result is now a trivial application of the chain rule to $p_i \circ s$:
\begin{proposition}\label{prop:diffProb13t}
Let $x_0 \in E_0$, $t \in (t_{N-1} , T)$. If $p_i$ is differentiable for all $i=1,\ldots,k$, then $G_\Wc$ is differentiable in $z$ over $E_0\times \Re_+ \times \Re$.
	\end{proposition}

} 
{
	\begin{proof}
	See full length technical report\cite{hya_gf_descentonline}.
	\end{proof}
	

We now address how Problem \eqref{eqHybridProg} might be solved. Functions $F$, $G_0$ and $G_\Wc$ are differentiable in $z$. It is not clear that $p_i(s_x(t);x_0)$, \emph{as a function of $z$}, is differentiable. We now show that under some assumptions on the $p_i$, for trajectories of linear systems, $p_i$ is in fact differentiable in both $x$ and $t$, over an appropriate range of $t$. This implies differentiability in $z$. Therefore, standard gradient-based optimizers can be used to solve Problem \eqref{eqHybridProg}. 
Due to lack of space, proofs are omitted. They can be viewed in the full-length report on-line\cite{hya_gf_descentonline}.
\todo[inline]{Put correct URL}

$s^{(i)}_x(\tau)$ will denote the point, at time $\tau$, on the trajectory starting at $x \in Inv(l_i)$, and evolving according to the dynamics of location $i$. 
$s_x(t)$ (without superscript) will denote the hybrid trajectory, traversing one or more locations. 
Recall that $E(x_0)$ is the ellipsoid of robustness centered at $x_0$. Then define $E_0 \triangleq \mathrm{int}(E(X_0) \bigcap X_0)$,
and $t_i^+ = min\{t | s(E_0, t) \subset Inv(l_i) \}$, 
where $s(E_0,t)$ is the image, under $s$ of $E_0$ at time $t$. $t_i^+$ is well-defined as shown in \cite{JuliusFALP07hscc}.

\begin{proposition} \label{prop:diffProb13x_manylocs}
Let $t \in [t_{N-1}^+ ,T]$, and consider the hybrid trajectory over $N \geq 1$ locations. If $s^{(0)}$ is differentiable in $x$ over $E_0$, the $s^{(i)}$, $i>0$, are differentiable in $x$ over $Inv(l_i)$, and $p_i(y;x_0)$ are differentiable in $x$ over $\Re^n$ for all $i$, then $s_x(t)$ is differentiable at $x_0$.
\end{proposition}

\begin{proposition} \label{prop:diffProb13t}
Let $x \in E_0$ and $t \in [t_{N-1}^+ , T)$, and consider the hybrid trajectory over $N \geq 1$ locations. Then $s_x(t)$ is differentiable in $t$ over $[t_{N-1}^+ , T)$.
	\end{proposition}

	\todo[inline]{diff in x and t $\Rightarrow$ diff in z}
	
}	 
	 
	 We choose Sequential Quadratic Programming (SQP), as a good general-purpose optimizer to solve Problem~\ref{eqHybridProg}. 
	 SQP is a Q-quadratically convergent iterative algorithm.
At each iterate, $G_\Wc(x_i,t_i,\nu_i)$ is computed by simulating the system at $x_i$. This is the main computational bottleneck of this method, and will be discussed in more detail in the Experiments section.

		\subsection{Convergence to a local minimum}\label{sec:iteratedProblem13}
	
	Solving Problem \eqref{eqHybridProg}, for a given $\Wc$, produces a descent direction for the robustness function. However, one can produce examples where a local minimum of $F(\cdot)$ is not a local minimum of the robustness function $f$. 	
	This section derives conditions under which \emph{repeated} solution of Problem \eqref{eqHybridProg} yields a local minimum of the robustness function.
	
	For $i=0,1,2,\dots$, let $x_i \in X_0 \bigcap E(x_{i-1})$, and let $t_i$ be the time when $s_{x_i}$ is closest to $\unsafe$. 
	Let $\Wc_i = P(x_i; t_i)$ be the descent set for this trajectory. 
	For each $\Wc_i$, one can setup the optimization Problem (\ref{eqHybridProg}) with $\Wc = \Wc_i$, and initial point $(x_i, t_i, 0)$; this problem is denoted by Prob\ref{eqHybridProg}[$\Wc_i$]. 
	(Recall from the proof of Proposition \ref{prop:hybDesc} that $\nu = 0$ at the initial point of the optimization problem).
	Finally, let $z_i^* = (x_i^*, t_i^*, \nu_i^*)$ be the minimum obtained by solving Prob\ref{eqHybridProg}[$\Wc_i$].
	
%
%

\begin{algorithm}[tb]
\caption{Robustness Ellipsoid Descent (RED)}
{\bf Input}: An initial point $x_0 \in X_0$, and corresponding $t_0$. \\
{\bf Output}: $z_Q$.
\label{alg:red}

\begin{algorithmic}[1]
\State Initialization: $i = 0$
\State Compute $z_i^*=(x_i^*, t_i^*, \nu_i^*)$ 
       = minimum of Prob\ref{eqHybridProg}[$\Wc_i$].
\While {$\nu_i^* < 0$}
 \State $x_{i+1} \leftarrow x_i^*$
 \State $t_{i+1} = \arg\min_t d_\unsafe(s_{x_{i+1}}(t))$
 \State $\Wc_{i+1} = P(x_{i+1})$ 
 \State Compute $z_i^*=(x_i^*, t_i^*, \nu_i^*)$ 
       = min of Prob\ref{eqHybridProg}[$\Wc_{i+1}$].
 \State $i=i+1$ 
\EndWhile\\
\State Return $z_Q \triangleq z_i^*$ 
 
\end{algorithmic}
\end{algorithm}		
	Algorithm~\ref{alg:red} describes how to setup a sequence of optimization problems that leads to a local minimum of $f$. It is called Robustness Ellipsoid Descent, or RED for short.
	
	\begin{proposition}\label{prop:REDterminates}
	Algorithm~\ref{alg:red} (RED) terminates
	\end{proposition}
	
	\begin{proof}
	Proposition \ref{prop:hybDesc} holds for each problem Prob\ref{eqHybridProg}[$\Wc_{i}$]. Therefore, each solution with $\nu_i < 0$ gives a trajectory $s_{x_i^*}$ with a smaller robustness than $s_{x_{i-1}^*}: f(x_i^*) < f(x_{i-1}^*)$. Thus $(f(x_i))_{i\in N}$ is a decreasing sequence, lower bounded by $0$. Therefore, it converges to a limit $r \geq 0$. But how to prove that this limit is indeed a minimum of f? 

	\end{proof}
	
	\begin{proposition}\label{prop:localMin}
	Assume that Algorithm \ref{alg:red} halts at a point $z_Q= (x_Q,t_Q,\nu_Q)$, for which there exist $\overline{t}_1, \overline{t}_2$ such that: 
	\begin{itemize}
	\item $0 \leq \overline{t}_1 \leq t_Q \leq \overline{t}_2$
	\item $d_\unsafe(s_{x_Q}(t)) > f(x_Q) \forall t \in T_R \triangleq [0, \overline{t}_1] \cup [\overline{t}_2, T)$, and 
	\item $\overline{t}_2-\overline{t}_1$ is `sufficiently small'. 
	\end{itemize}
	Then $x_Q$ is a local minimum of the robustness function $f$.	
	\end{proposition}
	
	\begin{proof}
	We assume that the trajectory starting at $x_Q$ is safe - otherwise, we're done since we found an unsafe trajectory.
		
	Two tubes will be constructed: one contains $s_{x_Q}$ over $(\overline{t}_1, \overline{t}_2)$, the other contains it over $T_R$. They are such that no trajectory in them gets closer to $\Wc_Q$ than $s_{x_Q}$. Then it is shown that all trajectories in a neighborhood of $x_Q$ are contained in these tubes, making $x_Q$ a local minimum of the robustness function $f$.	
	
\begin{figure}[ht]
\centering
\includegraphics[width=200pt]{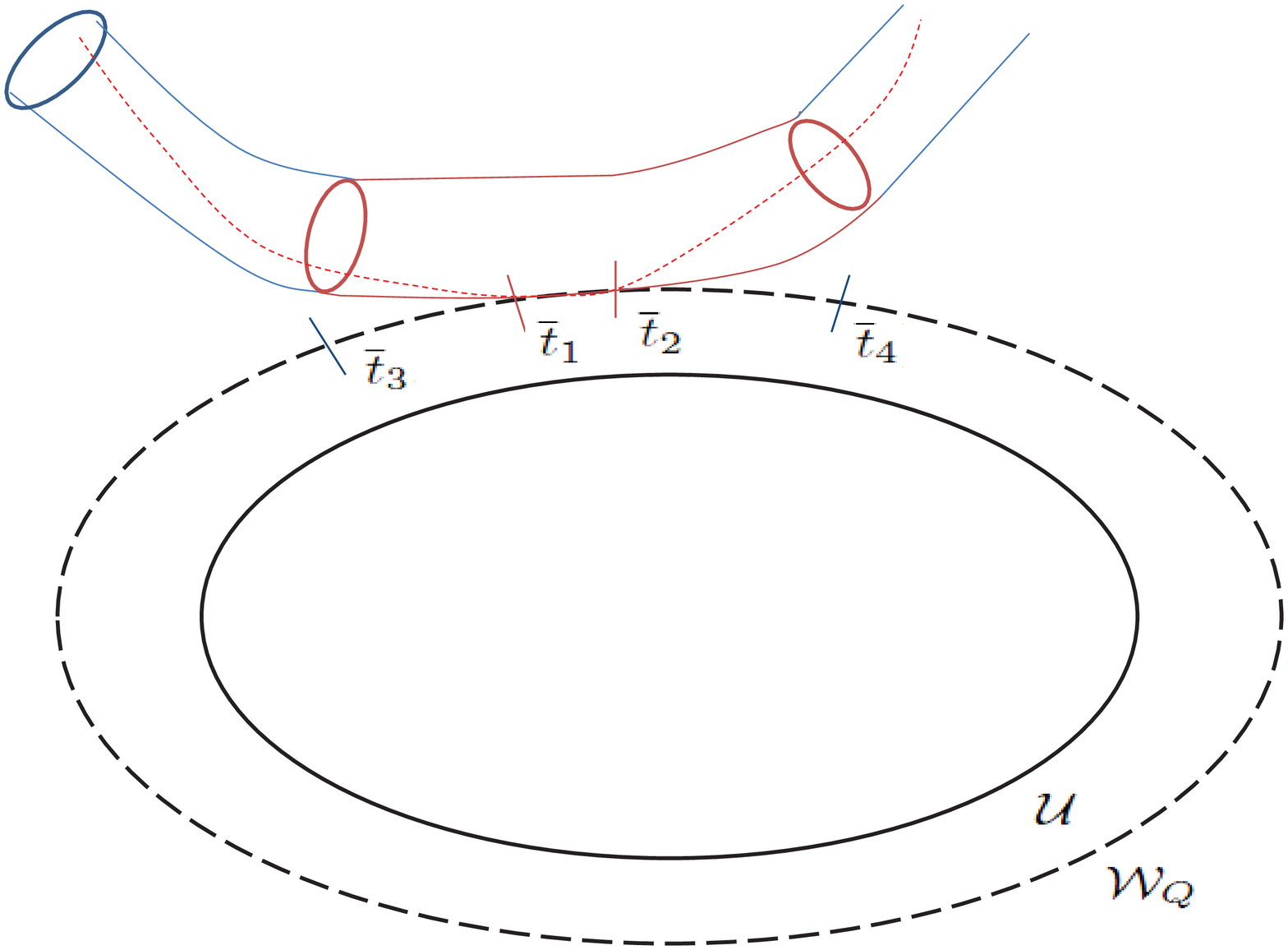}
\caption{[Proof of Prop.\ref{prop:localMin}] All trajectories starting in a neighborhood of $x_Q$ will be contained in the orange tube over $T_R$ and in the green tube over $(\overline{t}_1, \overline{t}_2)$}
\label{fig:tubes}
\end{figure}

	By the halting condition, $\nu_Q = 0$. Since the optimizer always returns a local minimum of the objective function $F$, there exists a neighborhood $N(z_Q)$ of $z_Q$ such that for all $z \in N(z_Q), F(z) \geq F(z_Q) = \nu_Q = 0$
	\[\Leftrightarrow \forall (x, t, \nu) \in N(z_Q), s_{x}(t) \notin \textrm{int}\Wc_Q\]
	\begin{equation*}\label{eq:ztube}
	\Leftrightarrow \forall (x, t, \nu) \in N(z_Q), d_\unsafe(s_{x}(t)) \geq f(x_Q)
	\end{equation*}		
	
	 $N(z_Q)$ can be expressed as 	
	\[N(z_Q) = B(x_Q, \epsilon)\times (\overline{t}_3, \overline{t}_4) \times (-\nu_l, \nu_l)\] 
	\[\epsilon >0, \nu_l > 0, B(x_Q,\epsilon) \subset E(x_0)\cap X_0\]		
	(Since $\Re^n \times \Re_+ \times \Re$ is a finite product, the box and product topologies are equivalent, so it doesn't matter which one we use.)	
	
	We now precise the notion of `small enough': we require that 
	\begin{equation*}
	(\overline{t}_1,\overline{t}_2) \subseteq (\overline{t}_3,\overline{t}_4)
	\end{equation*}
	
	Therefore 
	\[\forall x \in B(x_Q, \epsilon), t\in (\overline{t}_1, \overline{t}_2), d_\unsafe(s_{x}(t)) \geq f(x_Q)\]
	
	Thus
	\begin{equation}\label{eq:tube1}
	\forall x \in B(x_Q, \epsilon), \inf_{t\in (\overline{t}_1, \overline{t}_2)}d_\unsafe(s_{x}(t)) \geq f(x_Q)
	\end{equation}	
	 
	We now study the behavior of trajectories starting in $B(x_Q, \epsilon)$ over the remaining time periord $T_R$. Recall that $\unsafe \subset \Wc_Q$. Let $w^o$ be any point on the boundary $\partial \Wc_Q$. Then
	
	\begin{equation*}
	\begin{split}
	\forall t \in T_R, d_\unsafe(s_{x_Q}(t)) > f(x_Q) = d_\unsafe(w^o) > 0  \\
	\Rightarrow \forall t\in T_R, d_{\Wc_Q}(s_{x_Q}(t)) > 0 
	\end{split}
	\end{equation*}
	
	Then
	\[\Lambda = \inf\{d_{\Wc_Q}(s_{x_Q}(t))|t\in T_R\} > 0\]

	$s_x$ is continuous as a function of $x$ for every $t$, therefore
	\begin{equation*}
	\exists \delta >0 \textrm{ s.t. }x \in B(x_Q,\delta) \Rightarrow d(s_{x_Q}(t),s_x(t)) < \Lambda
	\end{equation*}			
	
	Pick any point $w \in \Wc_Q$. Then $\forall x \in B(x_Q,\delta)$ and $t\in T_R$
	\begin{equation*}
	\begin{split}	
	d(s_{x_Q}(t), w) & \leq d(s_{x_Q}(t), s_x(t)) + d(s_x(t),w)\\
									 & < \Lambda + d(s_x(t),w)\\
	\Rightarrow d(s_x(t),w) & > d(s_{x_Q}(t), w) - \Lambda 								 
	\end{split}
	\end{equation*}
	
	Minimizing both sides over $w \in \Wc_Q$,
	\begin{equation*}
	\begin{split}
	d_{\Wc_Q}(s_x(t)) > d_{\Wc_Q}(s_{x_Q}(t)) - \Lambda \geq 0 \\
	\Rightarrow \inf_{t\in T_R} d_{\Wc_Q}(s_x(t)) \geq 0 
	\end{split}
	\end{equation*}
	
	In conclusion
	\begin{equation}\label{eq:tube2}
	\forall x \in B(x_Q,\delta), \inf_{t\in T_R} d_\unsafe(s_x(t)) \geq f(x_Q)
	\end{equation}	
	
	Putting Eqs.\eqref{eq:tube1} and \eqref{eq:tube2} together, it comes that $\forall x \in B(x_Q, \min\{\epsilon,\delta\})$
	\begin{equation*}
	\begin{split}
	\inf_{t\in \Re_+} d_\unsafe(s_x(t)) \geq f(x_Q) \\
	\Leftrightarrow \forall x \in B(x_Q, \min\{\epsilon,\delta\}),f(x) \geq f(x_Q)
	\end{split}
	\end{equation*}	
	
	and $x_Q$ is a local minimum of the robustness $f$.	
	\end{proof}

	\subsection{Ellipsoid Descent with Stochastic Falsification}\label{subsecREDMC}
As outlined in the introduction, the proposed method can be used as a sub-routine in a higher-level stochastic search falsification algorithm. A stochastic search will have a ProposalScheme routine: given a point $x$ in the search space, ProposalScheme will propose a new point $x'$ as a falsification candidate. Robustness Ellipsoid Descent (RED) may then be used to further descend from some judiciously chosen proposals. Algorithm \ref{alg:red+mc} illustrates the use of RED within the Simulated Annealing (SA) stochastic falsification algorithm of \cite{NghiemSFIGP10hscc}. $U(0,1)$ denotes a number drawn uniformly at random over $(0,1)$. Given two samples $x$ and $y$, BetterOf($x,y$) returns the sample with smaller robustness, and its robustness. 

\begin{algorithm}[tb]
\caption{RED with Simulated Annealing (SA+RED)}

{\bf Input}: An initial point $x \in X_0$. \\
{\bf Output}: Samples $\Theta \subset X_0$. \\
{\bf Initialization}: BestSoFar = $x$, $f_b = f$(BestSoFar)
\label{alg:red+mc}

\begin{algorithmic}[1]
\While {$f(x) > 0$}
\State $x' = $ ProposalScheme($x$) 
\State $\alpha =$ exp $(-\beta(f(x')-f_b))$
\If {$U(0,1) \leq \alpha$}
	\State $x^* = $ RED($x'$)
	\State $x = x^*$ 	
\Else
	// Use the usual acceptance criterion
	\State $\alpha =$ exp $(-\beta(f(x')-f(x)))$
	\If {$U(0,1) \leq \alpha$}
		$x = x'$
	\EndIf	
\EndIf
\State (BestSoFar,$f_b$) = BetterOf($x$, BestSoFar)
\EndWhile 
\end{algorithmic}

\end{algorithm}

For each proposed sample $x'$, it is \emph{attempted} with certainty if its robustness is less than the smallest robustness $f_b$ found so far. Else, it is attempted with probability $e^{-\beta(f(x')-f_b)}$ (lines 3-4). If $x'$ is attempted, RED is run with $x'$ as starting point, and the found local minimum is used as final accepted sample (line 6).
If the proposed sample is not attempted, then the usual acceptance-rejection criterion is used: accept $x'$ with probability $\min\{1,e^{-\beta(f(x')-f(x))}\}$.
As in the original SA method, ProposalScheme is implemented as a Hit-and-Run sampler (other choices can be made). The next section presents experimental results on three benchmarks.
	
	\subsection{Experiments}\label{sec:Experiments}

This section describes the experiments used to test the efficiency and effectiveness of the proposed algorithm SA+RED, and the methods compared against it. 

We chose 3 navigation benchmarks from the literature: Nav0 (4-dimensional with 16 locations) is a slightly modified benchmark of~\cite{FehnkerI04}, and it is unknown whether it is falsifiable or not. Nav1 and Nav2 (4-dimensional with 3 locations) are the two hybrid systems in the HSolver library of benchmarks\cite{hsolverRatschan}, and are falsifiable. We also chose a filtered oscillator, Fosc (32-dimensional with 4 locations), from the SpaceEx library of benchmarks~\cite{FrehseCAV11}.
\ifthenelse{\boolean{REPORT}}
{We describe the Nav0 benchmark that we used, as it a slightly modified version of the benchmark in \cite{FehnkerI04}.

\begin{figure}[t]
\begin{center}
\includegraphics[width=8cm]{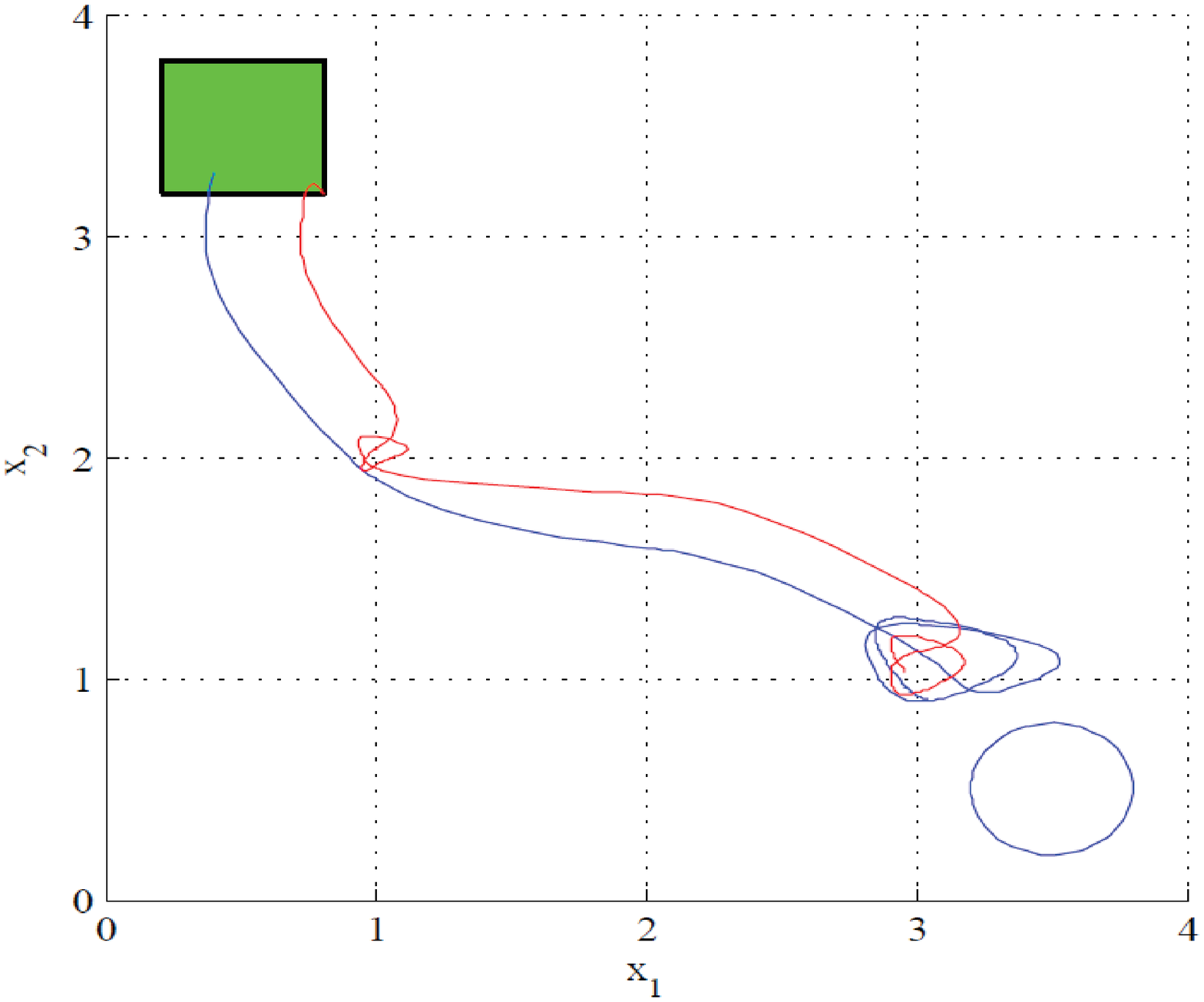}
\caption{The navigation benchmark example.} 
\label{fig:bench}
\end{center}
\end{figure}

\begin{exmp}[Navigation Benchmark \cite{FehnkerI04}]
The benchmark studies a hybrid automaton $\Hc$ with a variable number of discrete locations and 4 continuous variables $x_1$, $x_2$, $y_1$, $y_2$ that form the state vector $x = [x_1\; x_2\; y_1\; y_2]^T$. 
The structure of the hybrid automaton can be better visualized in Fig. \ref{fig:bench}. 
The invariant set of every $(i,j)$ location is an $1 \times 1$ box that constraints the position of the system, while the velocity can flow unconstrained. 
The guards in each location are the edges and the vertices that are common among the neighboring locations.

Each location has affine constant dynamics with drift. 
In detail, in each location $(i,j)$ of the hybrid automaton, the system evolves under the differential equation $\dot x = A x - B u(i,j)$ where the matrices $A$ and $B$ are
\[ A = \left [ \begin{smallmatrix}
0 & 0 & 1& 0 \\
0 & 0 & 0& 1 \\
0 & 0 & -1.2 & 0.1 \\
0 & 0 & 0.1 & -1.2
\end{smallmatrix} \right ]
\quad \mbox{ and } \quad
 B = \left [ \begin{smallmatrix}
0 & 0  \\
0 & 0  \\
 -1.2 & 0.1 \\
 0.1 & -1.2
\end{smallmatrix} \right ] \]
and the input in each location is 
$$u(i,j) = [\sin(\pi C(i,j)/4) \; \cos(\pi C(i,j)/4)]^T.$$ 
The array $C$ is one of the two parameters of the hybrid automaton that the user can control and it defines the input vector in each discrete location. 
Here, we consider the input array denoted in Fig. \ref{fig:bench}. 

The set of initial conditions is the set $H_0 = \{13\} \times [0.2 \; 0.8] \times [3.2 \; 3.8] \times [-0.4 \; 0.4]^2$ (green box in Fig. \ref{fig:bench}) 
and the unsafe set is $\unsafe = \{4\}  \times \{x \in \Re^4 \; | \; ||x-(3.5 \; 0.5 \; 0 \; 0)|| \leq 0.3\}$ (red circle in Fig. \ref{fig:bench}). This is slightly modified from the original benchmark to simplify the programming of the $p_i$ functions.
Sample trajectories of the system appear in \ref{fig:bench} for initial conditions $[0.8 \; 3.2 \; -0.2 \; 0.35]^T$ (red trajectory) and $[0.4 \; 3.3 \; -0.1 \; -0.1]^T$ (blue trajectory).
Note that the two trajectories follow different discrete locations.

 \label{exm:navbench}
\end{exmp} }
{}
The methods compared are: SA+RED, pure Simulated Annealing (SA) \cite{NghiemSFIGP10hscc}, mixed mode-HSolver (mm-HSolver) \cite{hsolverRatschan}, and the reachability analysis tool SpaceEx\cite{FrehseCAV11}. 
Because ours is a falsification framework, SpaceEx is used as follows: for a given bound $j$ on the number of discrete jumps, SpaceEx computes  an \emph{over}-approximation $\overline{R(j)}$ of the set reachable in $j$ jumps $R(j): R(j) \subset \overline{R(j)}$. If $\overline{R(j)} \cap \unsafe$ is empty, then \emph{a fortiori} $R(j)\cap \unsafe$ is empty, and the system is safe if trajectories are restricted to $j$ jumps. If, however, $\overline{R(j)} \cap \unsafe \neq \emptyset$, no conclusion can be drawn. 

Because SA and SA+RED are stochastic methods, their behavior will be studied by analyzing a number of runs. A \emph{regression} will mean a fixed number of runs, all executed with the same set of parameters, on the same benchmark. 
mm-HSolver is deterministic, and thus one result is presented for benchmarks Nav1 and Nav2 (Nav0 was not tested by mm-HSolver's authors \cite{hsolverRatschan}). The mm-HSolver results are those reported in the literature. 
SpaceEx was run in deterministic mode on Nav0 (specifically, we set parameter `directions' = `box'\cite{FrehseCAV11}).

{\bf{Parameter setting}}: We set the test duration $T = 12$sec, which we estimate is long enough to produce a falsifying trajectory for Nav0 if one exists. 
For SA+RED, we chose to generate 10 samples ($|\Theta|$ = 10). We will see that even this small number 	is enough for the algorithm to be competitive. 
A regression consists of 20 jobs.
The SpaceEx parameters were varied in such a way that the approximation $\overline{R}$ of the reachable set $R$ became increasingly precise. 
\ifthenelse{\boolean{REPORT}}
{
Clustering\% was given the values 0, 20 and 80 (the smaller the Clustering\%, the better the approximation and the longer the runtime). 
The ODE solver timestep $\delta$ was given the values $0.0008, 0.02, 0.041$ seconds. These are, respectively, the minimum, median, and average values of $\delta$ used by the variable step-size ODE solver used by SA+RED. The smaller $\delta$, the better the approximation and the longer the runtime. 
The following parameters were fixed: `directions' = `box', `Local time horizon' = 10sec, rel-err = abs-err = 1.0e-10. 
The Nav0 SpaceEx configuration files can be obtained by request from the authors.
}
{See \cite{hya_gf_descentonline}.}

{\bf{The performance metrics}}: Each run produces a minimum robustness. For a given regression, we measure: the smallest, the average, and the largest minimum robustness found by the regression (min, avg, max in Table \ref{tb:REDvsSA}). The standard deviation of minimum robustness is also reported ($\sigma_f$). For SpaceEx, we had to simply assess whether $\overline{R(j)}$ intersected $\unsafe$ or not.

{\bf{The cost metric}}: Each run also counts the number of simulated trajectories in the course of its operation: SA simulates a trajectory for each proposed sample, SA+RED simulates a trajectory each time the constraint function of Prob\ref{eqHybridProg}[$W_i$] is evaluated (and for each sample), and mm-HSolver simulates trajectories in falsification mode. 
The trajectories simulated by SA and SA+RED have a common, fixed, pre-determined duration $T$. Thus the cost of these algorithms can be compared by looking at the Number of Trajectories (NT) each simulates (column $\overline{NT}$ in Table \ref{tb:REDvsSA} - the overline denotes an average). 
The trajectories computed by mm-HSolver have varying lengths, determined by a quality estimate. So for comparison, we report the number of single simulation steps ($SS$), i.e. the number of points on a given trajectory (column $\overline{SS}$ - mm-HSolver, being deterministic, has one value of $SS$). Unfortunately, $SS$ doesn't include the cost of doing verification in mm-HSolver, so it should be considered as a \emph{lower bound} on its computational cost. On the other hand, because of the choice of $T$, the $SS$ numbers reported for SA+RED should be treated as \emph{upper bounds}: choosing a shorter a-priori $T$ will naturally lead to smaller numbers. An exact comparison of the costs of SA+RED and mm-HSovler would require knowing the duration of the shortest falsifying trajectory, and setting the a-priori $T$ to that, and somewhat incorporating the cost of verification.
The operations that SpaceEx does are radically different from those of the other methods compared here. The only way to compare performance is through the runtime. 

{\bf{Experimental setup}}: we impose an upper limit $NT_{MAX}$ on $NT$: SA+RED is aborted when its $NT$ reaches this maximum, and SA is made to generate $NT_{MAX}$ samples. (Of course, SA+RED might converge before simulating all $NT_{MAX}$ trajectories). 3 values were chosen for $NT_{MAX}$: 1000, 3000 and 5000. For each value, a regression is run and the results reported.  This allows us to measure the competitiveness of the 2 algorithms (i.e. performance for cost). 

{\bf{Experiments}}: Table~\ref{tb:REDvsSA} compares SA+RED to SA: we start by noting that SA+RED falsified Nav2, whereas SA failed to so. On most regressions, SA+RED achieves better performance metrics than SA, for the same (or lower) computational cost. 
This is consistent whether considering best case (min), average case (avg) or worst case (max). 
There are 2 exceptions: for Nav1 and Nav2, $NT_{MAX} = 5000$ produces better average and max results for SA than for SA+RED. 
When running realistic system models, trajectory simulation is the biggest time consumer, so effectively $NT$ is the limiting factor. So we argue that these 2 exceptions don't invalidate the superiority of SA+RED as they occur for high values of $NT$ that might not be practical with real-world models.
(In these cases, we observed that SA eventually produces a sequence of samples whose trajectories finish by making a large number of jumps between locations 3 and 2, with a relatively high robustness. From there SA then produces a sample with 0 (or close to 0) robustness. This happens on every Nav1 run we tried, and most Nav2 runs, resulting in the numbers reported. The RED step in SA+RED seems to avoid these trajectories by `escaping' into local minima, and is worthy of further study.)

\begin{table*}[t]
\centering
\begin{tabular}{|c|c|c|c|c|c|}
\hline
System & $NT_{MAX}$      & $\overline{NT}$ & $\sigma_f$ & SA+RED Rob.    & SA Rob.\\
       &                 & ($\sigma_{NT}$) &            & min, avg, max & min, avg, max  \\
\hline \cline{2-5}
Nav0 & 1000 & 1004 (1.4) &  0.022  & 0.2852, 0.30,0.35& 0.2853,0.33,0.33 \\ 
     & 3000 & 2716 (651) &  0.019   &0.2852,0.29,0.32&  0.2858,0.31,0.36  \\ 
     & 5000 & 4220 (802) &  0.009 & 0.285,0.28,0.32& 0.286,0.32,0.35 	\\ 
\hline \cline{2-5}
Nav1 & 1000 & 662  (399)  & 0.21  & 0,0.43,0.65& 0,0.96,1.88 \\ 
     & 3000 & 1129 (1033) & 0.23  & 0,0.39,0.65& 0,0.99,1.80 \\ 
     & 5000 & 1723 (1770) & 0.23  & 0,0.38,0.68& 0,0,0 \\ 
\hline \cline{2-5}
Nav2 & 1000 & 902  (246)  & 0.32 	& 0,0.54,0.78& 0.3089,1.11,1.90  \\ 
     & 3000 & 1720 (1032) & 0.3 	& 0,0.53,0.83&  0.3305,1.29,1.95 \\ 
     & 5000 & 1726 (1482) & 0.27  & 0,0.62,0.79& 0,0.002,0.01 \\ 
\hline \cline{2-5}
Fosc & 1000 & 1000  (9.3)  & 0.024 	& 0.162,0.206,0.251 & 0.1666,0.216,0.271  \\ 
     & 3000 & 3000 (8.7) & 0.024 	& 0.163,0.203,0.270&  0.173,0.212,0.254 \\ 
     & 5000 & 5000 (11) & 0.028  & 0.167,0.193,0.258& 0.185, 0.218, 0.245 \\       
\hline
\end{tabular}
\caption{Comparison of SA and SA+RED. To avoid clutter, Robustness values are reported to the first differing decimal, with a minimum of 2 decimals. $\sigma_f$ is standard deviation of robustness for SA+RED.}
\label{tb:REDvsSA}
\end{table*}
%

Table \ref{tb:resREDvsHSolver} compares SA+RED to mm-HSolver. We note that SA+RED falsifies the benchmarks, as does mm-HSolver. 
For Nav1, $\overline{SS}$ is greater than mm-HSolver's $SS$, though the falsifying runs have $SS$ values (last column) both smaller and larger than mm-HSolver. For Nav2, which appears to be more challenging, SA+RED performed better on average than mm-HSolver. However, we point out again that exact comparison is hard.

\begin{table*}[t]
\centering
\begin{tabular}{|c|c|c|c|c|c|c|}
\hline
System & $NT_{MAX}$ & $\overline{SS}$ & $\sigma_f$ & SA+RED Rob    & mm-HSolver      & $SS$ at min Rob   \\
       &            & ($\sigma_{SS}$) &            &min, avg, max  & Rob, $NT$, $SS$ & for SA+RED \\
\hline 
Nav1 & 1000         &  47k (30k) 	& 0.21 	 & 0,0.43,0.65 & 0, 22,5454           & 0,1560,16k \\ 
     & 3000 				&  79k (76k) 	&0.23  	 & 0,0.39,0.65  &                      &0,0,1600, 127k\\ 
     & 5000 				&  143k (141k)&0.23 	 & 0,0.38,0.68 &											 &7660, 38k, 102k, 159k\\ 
\hline 
Nav2 & 1000 				&  63k (18k) 	&0.32    & 0,0.54,0.78 & 0, 506, 138k          &2888, 74k\\ 
     & 3000 				&  126k (80k) & 0.3 	 & 0,0.53,0.83  &  											&14k,  57k, 210k\\ 
     & 5000 				&  124k (114k)& 0.27   & 0,0.622,0.79	&  											& 3450, 121k, 331k \\ 
\hline
\end{tabular}
\caption{Comparison of SA+RED and mm-HSolver. The last column shows some of the $SS$ values at which min robustness is achieved by SA+RED on various runs.}
\label{tb:resREDvsHSolver}
\end{table*}

For SpaceEx running on Nav0, we observed that our initial parameter set produces an $\overline{R(j)}$ that intersects $\unsafe$. Since this is inconclusive, we modified the parameters to get a better approximation. 
For parameter values (Clustering\%, $\delta$) = $(0,0.0008)$, $\overline{R(j)}$ and $\unsafe$ were almost tangent, but SpaceEx runtimes far exceeded those of SA+RED (more than 1.5 hours).
Moreover, SpaceEx did not reach a fixed point of its iterations (we tried up to $j = 200$ iterations before stopping due to high runtimes). Thus, we can not be sure that all of the reachable space was covered. While this may be seen as an analogous problem to the choice of $T$ in SA+RED, the computational cost of increasing $j$ is much more prohibitive than that of increasing $T$.
\ifthenelse{\boolean{REPORT}}
{
We now present some detailed runtime results. 
For SA+RED, `runtime' means the User time reported by the Unix \emph{time} utility. SA+RED was run on a dedicated Intel Xeon processor, x86-64 architecture, under the Unix OS. SpaceEx reports its own runtime. It was run on a Dual-Core Intel Centrino processor, under a Windows7 64b OS, with no other user applications running.

\begin{table*}[t]
\centering
\begin{tabular}{|c|c|c|c|c|c|}
\hline
Clustering\% &\multicolumn{3}{c|}{$\delta$(sec)} & SA+RED Runtime (sec) & $NT_{MAX}$ \\
             & 0.0008   & 0.002 & 0.041          & min,avg,max          &  \\
\hline 
80           & 737       &  30 	& 15 	           & 324, 426, 596        & 1000\\ 
20           & 1066      &  53 	& 33 	           & 620, 1132, 1385 			& 3000\\ 
10           & 1460      &  NA 	& NA	           & 767,1617, 2216 			& 5000\\  
0           & $> 5400$   &  NA 	& NA	           &  										&\\    
\hline
\end{tabular}
\caption{Comparison of SA+RED and SpaceEx runtimes. NA means the experiment was not run, because a more accurate run was required. The right-most columns shows the $NT_{MAX}$ constraint for which the SA+RED runtimes were obtained.}
\label{tb:resREDvsSpaceEx}
\end{table*}

%
} 
{Details can be found in the technical report\cite{hya_gf_descentonline}.}
Thus we may conclude that stochastic falsification and reachability analysis can play complementary roles in good design practice: first, stochastic falsification computes the robustness of the system with respect to some unsafe set. Guided by this, the designer may make the system more robust, which effectively increases the distance between the (unknown) reachable set and the unsafe set. Then the designer can run a reachability analysis algorithm where coarse over-approximations can yield conclusive results.


	\section{Conclusions}\label{secConclusions}

The minimum robustness of a hybrid system is an important indicator of how safe it is. 
In this paper, we presented an algorithm for computing a local minimum of the robustness for a certain class of linear hybrid systems. 
The algorithm can also be used to minimize the robustness of non-hybrid linear dynamic systems. 
When integrated with a higher-level stochastic search algorithm, the proposed algorithm has been shown to perform better than existing methods on literature benchmarks, and to complement reachability analysis.
We will next deploy this capability to perform local descent search for the falsification of arbitrary linear temporal logic specifications, not only safety specifications.
This investigation opens the way to several interesting research questions.
Most practically, reducing the number of tests $NT$ results in an immediate reduction of the computation cost. 
Also useful, 	is the determination of an appropriate test duration $T$, rather than a fixed arbitrary value.
 
In terms of performance guarantees, obtaining a lower bound on the optimum achieved in Problem \ref{eqHybridProg} could lead to a lower bound on the optimal robustness.
One level higher in the algorithm, it is important to get a theoretical understanding of the behavior of the Markov chains iterated by SA+RED to further improve it.

%
%
%
%
%
		
\bibliographystyle{splncs}
\bibliography{fainekos_bibrefs}

\end{document}